\documentclass[runningheads]{llncs}

\usepackage[T1]{fontenc}

\usepackage[utf8]{inputenc}
\usepackage[dvipsnames]{xcolor}
\usepackage{mathtools}
\usepackage{complexity}
\usepackage{adjustbox}
\usepackage{array}
\usepackage{multirow}
\usepackage{amssymb}

\usepackage{tikz}
\usetikzlibrary{fit,arrows,arrows.meta,calc}
\pgfdeclarelayer{background}
\pgfdeclarelayer{foreground}
\pgfsetlayers{background,main,foreground}

\usepackage[linesnumbered,ruled,commentsnumbered,longend]{algorithm2e}

\usepackage[numbers,sort&compress]{natbib}

\usepackage[nopatch=footnote]{microtype}

\usepackage{hyperref}
\hypersetup{
  colorlinks = true,
  citecolor = blue,
  linkcolor = blue,
  urlcolor = blue
}

\usepackage{lineno}

\usepackage{footnotebackref}

\usepackage{orcidlink}

\usepackage{array}
\usepackage{booktabs}

\usepackage[]{todonotes}

\DeclareMathOperator{\Var}{Var}
\DeclareMathOperator{\Lit}{Lit}
\DeclareMathOperator{\Aut}{Aut}

\newcommand{\F}{F}
\newcommand{\symm}{\sigma}
\newcommand{\assn}{\tau}
\newcommand{\wit}{\omega}
\newcommand{\ol}[1]{\overline #1}

\newcommand{\rat}{\textsc{rat}}
\newcommand{\drat}{\textsc{drat}}

\newcommand{\pr}{\textsc{pr}}

\newcommand{\sr}{\textsc{sr}}
\newcommand{\dsr}{\textsc{dsr}}
\newcommand{\lsr}{\textsc{lsr}}

\newcommand{\dsrtrim}{\textsc{dsr-trim}}
\newcommand{\lsrcheck}{\textsc{lsr-check}}

\newcommand{\veripb}{\textsc{VeriPB}}

\newcommand{\satsuma}{\textsc{satsuma}}
\newcommand{\dejavu}{\textsc{dejavu}}
\newcommand{\nauty}{\textsc{nauty}}
\newcommand{\Traces}{\textsc{Traces}}
\newcommand{\bliss}{\textsc{bliss}}
\newcommand{\breakid}{\textsc{BreakID}}

\newcommand{\cadical}{\textsc{CaDiCaL}}

\usepackage{changebar}

\begin{document}

\title{Orbitopal Fixing in SAT}

\author{Markus Anders\inst{1}\orcidlink{0009-0004-5992-8433}
    \and Cayden Codel\inst{2}\orcidlink{0000-0003-3588-4873}
    \and Marijn J. H. Heule\inst{2}\orcidlink{0000-0002-5587-8801}}

\authorrunning{M. Anders et al.}

\institute{
    RPTU Kaiserslautern-Landau, Kaiserslautern, Germany
    \\ \email{anders@cs.uni-kl.de}
    \and 
    Carnegie Mellon University, Pittsburgh, PA, United States
    \\ \email{\{ccodel,mheule\}@cs.cmu.edu}
}

\maketitle

\setcounter{footnote}{0}

\begin{abstract}
Despite their sophisticated heuristics,
boolean satisfiability (SAT) solvers are still vulnerable to symmetry,
causing them to visit search regions
that are symmetric to ones already explored.
While symmetry handling is routine in other solving paradigms,
integrating it into state-of-the-art proof-producing SAT solvers is difficult:
added reasoning must be fast,
non-interfering with solver heuristics,
and compatible with formal proof logging.
To address these issues,
we present a practical static symmetry breaking approach based on \emph{orbitopal fixing},
a technique adapted from mixed-integer programming.
Our approach adds only \emph{unit clauses},
which minimizes downstream slowdowns,
and it emits succinct proof certificates in the substitution redundancy proof system.
Implemented in the \satsuma{} tool,
our methods deliver consistent speedups on symmetry-rich benchmarks with negligible regressions elsewhere. 
\end{abstract}

\section{Introduction}

Boolean satisfiability (SAT) solvers power a wide range of industrial and academic applications~\cite{handbook}.
Yet despite decades of innovation,
state-of-the-art SAT solvers still lack robust, broadly deployed mechanisms for symmetry reasoning,
even though such mechanisms are commonplace in other paradigms~\cite{PfetschR19,GentPP06}.
Without explicit symmetry reasoning,
solvers can waste significant amounts of time
exploring search regions that are isomorphic to ones already ruled out,
leading to substantial slowdowns on highly symmetric instances.

To address this problem,
prior work in SAT has explored both preprocessing (static) and on-the-fly (dynamic) symmetry-breaking techniques~\cite{CrawfordGLR96,AloulMS03,DevriendtBBD16,JunttilaKKK20,Sabharwal09,DevriendtBCDM12,Devriendt0B17,AndersBR24}. 
The most used approach in SAT is \emph{static symmetry breaking},
which adds constraints to the formula before solving to avoid isomorphic solutions.
For example, in graph coloring, one can fix the color of a designated vertex, since any valid coloring can be permuted accordingly. 
Static methods are attractive in practice because their overhead is often modest~\cite{AndersBR24}.

Although these techniques can yield substantial speedups on highly symmetric formulas,
they can also incur severe regressions elsewhere.
A key culprit of this slowdown is overly aggressive symmetry breaking:
Adding too many clauses to the formula
often causes the solver's performance to degrade,
especially when the formula is \emph{satisfiable}.
(For instance, see work by Aloul et al.\ \cite{AloulMS03}.)
Overall, symmetry handling techniques must strike a delicate balance between
reasoning strength, 
computational cost,
and minimal interference with solver heuristics.


Complicating this trade-off even further,
SAT symmetry-breaking techniques must also be compatible with proof production.
This is because
modern SAT solvers (since 2016) are \emph{certifying algorithms}~\cite{McConnellMNS11},
meaning that they emit formally checkable proofs that their answers are correct.
Any additional symmetry reasoning must therefore integrate cleanly with proof generation and verification.

Today,
practical SAT symmetry-breaking tools suffer from several drawbacks.
All current tools produce structured lex-leader constraints~\cite{CrawfordGLR96},
which can blow up the size of the formula
and degrade learned clause quality
when encoded into SAT.
Proof logging is also problematic.
Proof logging for practical symmetry-breaking tools was only introduced very recently
by means of the dominance rule~\cite{BogaertsGMN23}.
This approach has received notable success,
with an implementation in \breakid~\cite{DevriendtBCDM12}
earning a special prize at SAT Competition 2023.
But while dominance-based rules are very general,
the new proof systems needed to support them are complicated to implement,
and their proofs are slow to check.

Interestingly,
some symmetry-breaking techniques for mixed-integer programming (MIP)
avoid the problems of using large lex-leader constraints
by applying symmetry reasoning in a more surgical manner.
For formulas that exhibit so-called row symmetry,
\emph{orbitopal fixing}~\cite{KaibelPeinhardtPfetsch11} 
breaks symmetries by adding only \emph{unit clauses}.
This is accomplished by combining insights on symmetry
with insights on cardinality.
Adapting such a technique to SAT should have 
far fewer downsides than introducing long, structured constraints,
such as lex-leader constraints.

\paragraph{Contribution.}
To tackle the challenges discussed above,
we adapt orbitopal fixing from MIP to SAT
to introduce three new methods of practical symmetry handling.
All of our methods follow three guiding principles:
\vspace{-5pt}
\begin{enumerate}
    \item They exclusively add \emph{unit clauses} to the formula.
    \item They simultaneously exploit \emph{symmetry} and \emph{cardinality}.
    \item They generate succinct proof certificates
          in the substitution redundancy (\sr{}) proof system~\cite{GochtN21,BussHierarchy},
          without the need for dominance-based rules.
\end{enumerate}

We implement our new techniques in the state-of-the-art symmetry breaking tool \satsuma~\cite{AndersBR24}.\footnote{\url{https://github.com/markusa4/satsuma}.}
Despite the apparent restrictions\textemdash foregoing lex-leader constraints and feature-rich proof systems\textemdash it turns out that, indeed, our approach produces strong practical results:
\vspace{-5pt}
\begin{enumerate}
    \item The performance of the state-of-the-art SAT solver \cadical{}~\cite{Cadical2024} is substantially improved on the 
          SAT Competition 2025, the SAT anniversary track of 2022, and a set of highly symmetric crafted benchmarks.
    \item The preprocessing overhead is negligible (less than $1\%$ of average solve time).
    \item The performance regression on \emph{satisfiable instances} is significantly smaller than for lex-leader constraints
          (even though lex-leader constraints achieve overall better pruning than our techniques on \emph{unsatisfiable} instances).
    \item The \sr{} proofs are succinct, easy to generate, and efficient to check.
\end{enumerate}

Overall,
our techniques offer a more lightweight, surgical, and stable approach to SAT symmetry breaking
than lex-leader constraints,
and our techniques can be easily combined with other symmetry handling methods.


\section{Preliminaries}

We assume that the reader is generally familiar with concepts from SAT solving.
For a broad introduction to the topic,
see the Handbook of Satisfiability~\cite{handbook}.

The propositional formulas we consider in this paper
are all in \emph{conjunctive normal form} (CNF),
meaning that they are conjunctions of disjunctive \emph{clauses}
containing \emph{literals}.
A literal~$\ell$ is either
a variable~$v$
or its negation~$\overline{v}$. 
In this paper,
we interpret clauses and formulas as sets.
For example,
we sometimes write the clause $(x \vee \overline{y} \vee z)$ as $\{ x, \overline{y}, z \}$.
Let $\Var(\F)$ and $\Lit(\F)$ be the set of variables and literals occurring in $\F$,
respectively.

Two formulas $\F$ and $\F'$ are \emph{equisatisfiable}
if $\F$ is satisfiable iff $\F'$ is satisfiable.
This definition is bidirectional,
but since we only consider formulas~$\F'$
that are formed by adding clauses to $\F$
(i.e., $\F \subseteq \F'$),
the reverse direction is trivial,
and thus we omit it in our proofs.

\subsection{Unique Literal Clauses}

A clause $C \in \F$ is a \emph{unique literal clause} (ULC)
with respect to $\F$
if none of its literals $\ell \in C$ appear elsewhere in $\F \setminus C$.
ULCs enjoy the following property,
which is key to adapting orbitopal fixing to SAT (see Section~\ref{subsec:orbitopal}).

\begin{lemma}[\cite{ShengRH25}, Lemma 4]
    Let $\F$ be a formula,
    and let $C \in \F$ be a ULC.
    If $\F$ is satisfiable, 
    then it can be satisfied by a truth assignment that sets exactly one literal in $C$ to true. 
    \label{lem:ulc}
\end{lemma}

Another nice property of ULCs is that the set of ULCs in a formula~$F$ can be computed in linear time:
First store how many times each literal appears in $\F$,
and then check each clause to see if all of its literals appear exactly once in $\F$.

\subsection{Syntactic Symmetry of Formulas} \label{sec:symmetries}

A symmetry~$\symm$ of a formula $\F$ is a permutation of $\Lit(\F)$
that maps $\F$ to itself.
Formally, let $\symm$ be a permutation of $\Lit(\F)$,
and define $\symm(\F)$ as
the formula created by relabeling the literals of $\F$ under $\symm$.
Then $\symm$ is a \emph{syntactic symmetry} of $\F$ if:
\begin{enumerate}
    \item $\symm(\F) = \F$, and
    \item $\neg \symm(l) = \symm(\ol{\ell})$ for all $\ell \in \Lit(\F) \quad$
    (i.e., $\symm$ commutes with negation).
\end{enumerate}

When defining a symmetry~$\symm$,
condition (2) says it is sufficient to specify $\symm(\ell)$ for only positive literals~$\ell$.
We will write symmetries as $\symm \coloneqq ( \ell \mapsto \ell', \dots)$,
meaning that $\symm (\ell) = \ell'$.
All literals not explicitly listed are assumed to map back to themselves,
i.e.,
$\symm (x) = x$.

\begin{figure}[t]
\tikzset{
  var/.style={
    draw,
    circle,
    fill=orange!30,
    minimum size=5mm, 
    inner sep=0pt,
  },
  clause/.style={
    draw,
    circle,
    fill=teal!30,
    minimum size=5mm, 
    inner sep=0pt,
  }
}
\centering
    \begin{tikzpicture}[scale=1.1]
    \foreach \name/\angle in {x/90, y/210, z/330}{
        \node[var] (\name) at (\angle:0.7) {$\name$};
        \node[var] (n\name) at (\angle:1.5) {$\overline\name$};
        \node[] (i\name) at (\angle:2.5) {};
    }
    \foreach \name/\angle in {a/150, b/270, c/30}{
        \node[clause] (\name) at (\angle:1.2) {$\vee$};
    }
    \node[clause] (d) at (0,0) {$\vee$};

    \draw (x) -- (nx);
    \draw (y) -- (ny);
    \draw (z) -- (nz);

    \draw (a) -- (nx);
    \draw (a) -- (ny);

    \draw (b) -- (ny);
    \draw (b) -- (nz);

    \draw (c) -- (nx);
    \draw (c) -- (nz);

    \draw (d) -- (x);
    \draw (d) -- (y);
    \draw (d) -- (z);

    \draw[-latex,thick,color=green!50!black] ($ (ix)!0.25!(iy) $) to[bend right] ($ (ix)!0.75!(iy) $);
    \draw[-latex,thick,color=green!50!black] ($ (iy)!0.25!(iz) $) to[bend right] ($ (iy)!0.75!(iz) $);
    \draw[-latex,thick,color=green!50!black] ($ (iz)!0.25!(ix) $) to[bend right] ($ (iz)!0.75!(ix) $);
\end{tikzpicture} \hspace{0.5cm}
    \begin{tikzpicture}[scale=1.1]
    \foreach \name/\angle in {x/90, y/210, z/330}{
        \node[var] (\name) at (\angle:0.7) {$\name$};
        \node[var, fill = purple!30] (n\name) at (\angle:1.5) {$\overline\name$};
        \node[] (i\name) at (\angle:2.5) {};
    }
    \foreach \name/\angle in {a/150, b/270, c/30}{
        \node[clause] (\name) at (\angle:1.2) {$\vee$};
    }
    \node[clause, fill=green!50!black!30] (d) at (0,0) {$\vee$};

    \draw (x) -- (nx);
    \draw (y) -- (ny);
    \draw (z) -- (nz);

    \draw (a) -- (nx);
    \draw (a) -- (ny);

    \draw (b) -- (ny);
    \draw (b) -- (nz);

    \draw (c) -- (nx);
    \draw (c) -- (nz);

    \draw (d) -- (x);
    \draw (d) -- (y);
    \draw (d) -- (z);

    \phantom{\draw[-latex,thick,color=green!50!black] ($ (ix)!0.25!(iy) $) to[bend right] ($ (ix)!0.75!(iy) $);}
    \phantom{\draw[-latex,thick,color=green!50!black] ($ (iy)!0.25!(iz) $) to[bend right] ($ (iy)!0.75!(iz) $);}
    \phantom{\draw[-latex,thick,color=green!50!black] ($ (iz)!0.25!(ix) $) to[bend right] ($ (iz)!0.75!(ix) $);}
\end{tikzpicture}
\caption{A graph modeling the symmetries of the formula $(x \vee y \vee z)
    \land (\ol x \vee \ol y)
    \land (\ol x \vee \ol z)
    \land (\ol y \vee \ol z)$.
    Every clause is connected to its component literals,
    and every literal is adjacent to its negation.
    On the left, the green arrows indicate the symmetry mapping $x$ to $y$, $y$ to $z$, and $z$ to $x$.
    On the right, the colors indicate the orbits of the vertices.}
    \label{fig:modelgraph}
\end{figure}
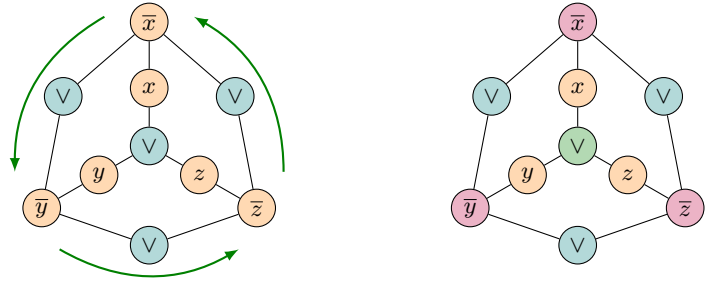

\begin{example} \label{ex:sym}
    Let $F = 
        (x \vee y \vee z)
        \land (\ol x \vee \ol y)
        \land (\ol x \vee \ol z)
        \land (\ol y \vee \ol z)$.
    Then the permutation
    $\symm := (x \mapsto y, y \mapsto z, z \mapsto x)$ is a symmetry of $\F$,
    since
    \begin{equation*}
        \symm(\F) =
        (y \vee z \vee x)
        \land (\ol y \vee \ol z)
        \land (\ol y \vee \ol x)
        \land (\ol z \vee \ol x)
        = \F.
    \end{equation*}
\end{example}

In practice, the symmetries of a CNF formula are computed
by modeling the formula as a graph
and then giving that graph to an off-the-shelf graph isomorphism solver,
such as
\nauty{}~\cite{practicaliso1, McKayP14},
\bliss{}~\cite{JunttilaK07,JunttilaK11}, 
\Traces{}~\cite{corr/abs-0804-4881,McKayP14}, 
or \dejavu{}~\cite{AndersS21,AndersS21b}.
Figure~\ref{fig:modelgraph} (left) illustrates a graph modeling the symmetries of Example~\ref{ex:sym}.
More-compact graph representations are typically used in practice~\cite{AloulRMS03}.

As it turns out,
the symmetries of a formula~$\Aut(\F)$\footnote{
    The notation $\Aut()$ is due to how the symmetries of $\F$ are also its \textbf{aut}omorphisms.
}
form a \emph{permutation group},
which means
we can use concepts from group theory to reason about them.
In this paper,
we use two such concepts:
stabilizers and orbits.

\emph{Stabilizers} are sets of symmetries that map literals back to themselves in certain ways.
The \emph{pointwise stabilizer}~$\Aut(\F)_{(L)}$
contains all symmetries of $\F$
that stabilize each individual literal in a set of literals~$L$,
while the \emph{setwise stabilizer}~$\Aut(\F)_{\{L\}}$
contains all symmetries of $\F$
that map $L$ back to itself.
Formally,
\begin{align*}
    \Aut(\F)_{(L)} & \coloneqq \{\symm \in \Aut(\F) \;|\; \symm(\ell) = \ell \text{ for all } \ell \in L \}, & \text{(Pointwise)} \\
    \Aut(\F)_{\{L\}} & \coloneqq \{\symm \in \Aut(\F) \;|\; \symm(L) = L\}. & \text{(Setwise)}
\end{align*}
Since the condition for pointwise stabilizers is stronger than the one for setwise stabilizers,
we have that $\Aut(\F)_{(L)} \subseteq \Aut(\F)_{\{L\}}$.
Each set is always nonempty,
since the identity permutation is a member of both stabilizers.

The \emph{orbit} of a literal~$\ell$ is the set of literals that can be reached from $\ell$ by a permutation group~$G$ of symmetries of $F$.
Often,
$G = \Aut(\F)$,
but $G$ is allowed to be any subgroup of $\Aut(\F)$.
Two literals $\ell_1$ and $\ell_2$ are \emph{in the same orbit}
with respect to $G$
if there exists a symmetry $\symm \in G$
such that $\symm(\ell_1) = \ell_2$.
The orbits under $G$ form a partition of the literals,
where literals in the same orbit are in the same equivalence class.
Figure~\ref{fig:modelgraph} (right) illustrates the orbits of Example~\ref{ex:sym}.

For a more general introduction to permutation groups,
we refer the reader to work by Seress~\cite{seress_2003}.

\subsection{Substitution Redundancy Proofs} \label{sec:prelims::subsec:sr}

When performing static symmetry breaking for SAT,
we add \emph{symmetry-breaking clauses} to a CNF formula~$F$
to forbid symmetric solutions.
For the additions to be valid,
we must prove that each addition preserves the equisatisfiability of $\F$.
We can write such a proof of equisatisfiability
in a \emph{clausal proof system},
and we can check the proof with a formally-verified proof checker.
In this paper,
we use the \emph{substitution redundancy} (\sr) proof system~\cite{GochtN21,RebolaPardo23,BussHierarchy},
which is a generalization of the popular
\rat{}~\cite{RAT12}
and \pr{}~\cite{PRTheory}
proof systems.

In a clausal proof system,
each proof step either adds a clause or deletes a clause.
Added clauses~$C$ must be \emph{redundant},
meaning that $\F$ and $(\F \land C)$ are equisatisfiable.
Addition steps may also include a \emph{witness}~$\wit$
that helps prove that $C$ is redundant.
Crucially,
these witnesses allow for efficient proof checking.

In \sr,
the witness is a substitution~$\wit : \text{Var}(\F) \rightarrow \text{Lit}(\F) \cup \{ \top, \bot \}$
that maps each variable to a literal or to a fixed truth value.\footnote{
    \pr{} witnesses are partial assignments, and \rat{} witnesses are partial assignments on a single literal.
    Thus, \sr{} is a natural generalization of these two systems.
}
This substitution extends to literals in the natural way.
As it turns out,
substitutions can express the symmetry reasoning
involved in orbitopal fixing,
which makes it easy to generate \sr{} proofs
for the symmetry-breaking clauses we discuss in this paper.

To show that $C$ is redundant,
it is sufficient to show that $C$ is \emph{substitution redundant} (\sr{}) for $F$.
In the definition below,
we use the following notation.
We write $\lnot C$ for the negation of the disjunctive clause~$C$,
i.e.,
$\lnot C \coloneqq \wedge_{\ell \in C}\, \ol \ell$.
We write $F_{| \wit}$
for the reduction of the formula~$F$ under the substitution~$\wit$,
where every literal~$\ell$ in $F$ is replaced with $\wit(\ell)$.
We write $\vdash_1$ for entailment via unit propagation,
where $F \vdash_1 \bot$ means that $F$ causes a contradiction under unit propagation,
$F \vdash_1 C$ means that $F \land \lnot C \vdash_1 \bot$,
and $F \vdash_1 G$ means that
$F \land \lnot D \vdash_1 \bot$ for all $D \in G$.

\begin{definition}[Substitution redundant]
    A clause~$C$ is \emph{substitution redundant}
    for a formula~$\F$
    if there exists a substitution~$\wit$
    such that $\F \wedge \lnot C \vdash_1 (\F \wedge C)_{|\wit}$.\label{def:sr}
\end{definition}

Intuitively,
the witness~$\wit$ provides a way to repair any assignment~$\assn$
that satisfies $F$ but not $C$ into an assignment that satisfies both.
If $\wit$ expresses a symmetry of $\F$,
then it suffices to show that the repaired assignment $\assn \circ \wit$ satisfies $C$,
where ``$\circ$'' acts as a kind of function composition,
with $(\assn \circ \wit) (\ell) = \wit (\ell)$ if $\wit (\ell) \in \{ \top, \bot \}$
and $(\assn \circ \wit) (\ell) = \assn (\wit (\ell))$ otherwise.

The \sr{} rule uses unit propagation $\vdash_1$
rather than general entailment $\vDash$
because 
the use of unit propagation enables the \sr{} rule to be checked efficiently by proof checkers.
Today,
only the \dsr{}/\lsr{}~\cite{CodelAH24} and \veripb{}~\cite{GochtN21} proof formats support \sr{} reasoning.
Our tool can generate proofs in either format.

\begin{example}\label{ex:sr-sec-php}
Consider the pigeonhole problem (PHP) of placing $m$ pigeons into $n$ holes
such that each pigeon gets its own hole.
Whenever $m > n$,
this task is impossible.
A common SAT encoding of PHP is:
\begin{equation*}
    \textsc{php} (m, n) = \bigwedge_{j = 1}^{m} \left(\bigvee_{i = 1}^{n} p_{i, j} \right)
    \, \wedge \,
    \bigwedge_{i = 1}^{n} \, \bigwedge_{1 \leq j < k \leq m}  \left( \ol{p}_{i, j} \lor \ol{p}_{i, k} \right),
\end{equation*}
where the variables~$p_{i, j}$ mean that pigeon~$j$ is placed in hole~$i$.
When visualized as a matrix,
the $m$ columns contain the at-least-one constraints,
and the $n$ rows contain the at-most-one constraints.

This encoding exhibits a lot of symmetry.
In particular,
we are free to relabel the holes or the pigeons however we wish.
(In other words, the encoding exhibits \emph{row symmetry}; see Section~\ref{subsec:orbitopal}.)
The presence of this symmetry allows us to use \sr{} reasoning to add redundant clauses to the formula.

Suppose we want to use the \sr{} rule to show that pigeon~1 does not go in hole~$1$,
i.e.,
that the unit clause $C = \{ \ol{p}_{1, 1} \}$ is \sr.
Since the left-hand side of the $\vdash_1$ turnstile in the \sr{} rule assumes $\lnot C$,
we are essentially assuming that pigeon~1 gets placed in hole~1.
To ``repair'' this situation,
we will use the witness~$\wit$ that swaps holes 1 and 2,
meaning that pigeon~1 now gets placed in hole~2.
Formally,
$\wit \coloneqq ( p_{1, 1} \mapsto \bot,\, p_{2, 1} \mapsto \top,\, p_{1, j} \mapsto p_{2, j},\, p_{2, j} \mapsto p_{1, j} )$.
Note that $\wit$ explicitly sets the truth values for $p_{1, 1}$ and $p_{2, 1}$,
which forces pigeon~1 to be placed in hole~2.
All other variables for holes 1 and 2 get swapped.
Viewing the variables as a matrix,
$\wit$ swaps rows 1 and 2.

We now show that $C$ and $\wit$ satisfy the \sr{} condition.
The good news is that most clauses in $(F \land C)_{| \wit}$ have a trivial unit propagation refutation.
In general,
any clause $D \in F \land C$
where $D_{| \wit} = \top$ or $D_{| \wit} = D$ has a trivial refutation.
Here,
the at-least-one column constraint containing $p_{2, 1}$
and any at-most-one row constraints containing $\ol{p}_{1, 1}$ are satisfied by $\wit$,
and the remaining column constraints and the constraints for rows 3 through $n$ are mapped back to themselves under $\wit$.
That leaves the at-most-one constraints for rows 1 and 2.

The row 1 constraints are easy.
Either they contain $\ol{p}_{1, 1}$ and are satisfied by $\wit$,
or they are mapped to a row 2 constraint,
which causes a trivial refutation.

Finally,
for the row 2 constraints,
we do some unit propagation.
By assuming $\lnot C = \{ p_{1, 1} \}$ on the left-hand side of $\vdash_1$,
we can derive $\{ \ol{p}_{1, j} \}$ for all $j \neq 1$ via unit propagation on $\{ \ol{p}_{1, 1}, \ol{p}_{1, j} \}$.
This lets us derive a refutation with the row 2 constraints on the right-hand side of $\vdash_1$,
since any constraint $\{ \ol{p}_{2, k}, \ol{p}_{2, k'} \}$,
mapped under $\wit$
will contain $\ol{p}_{1, j}$ for some $j$,
conflicting with the $\{ \ol{p}_{1, j} \}$ unit we derived.
\end{example}

\section{Fixing Rules}

In this section, we introduce our new symmetry-breaking techniques.
Notably,
our techniques exclusively \emph{assign} or \emph{fix variables},
i.e.,
they exclusively add unit clauses to the formula.
We also prove that each technique preserves equisatisfiability
and is compatible with \sr{} proof production.

\subsection{Orbitopal Fixing}\label{subsec:orbitopal}

Our first symmetry-breaking technique is \emph{orbitopal fixing},
which combines row symmetry in a matrix of literals~$M$
with the presence of unique literal clauses (ULCs) in the formula
to fix literals of $M$.
Our technique is inspired by a procedure of the same name used in 
MIP \cite{KaibelPeinhardtPfetsch11,mexi2025}.

Intuitively,
a subset of a formula's literals exhibit row symmetry
if they can be arranged into a rectangular matrix~$M$
such that there exist symmetries
that swap any two rows of $M$.
Formally,
let $F$ be a formula,
and let $M := (\ell_{i, j})$ be an $n \times m$ matrix comprising a subset of $\Lit(\F)$.
Then $M$ exhibits \emph{row symmetry}~\cite{FlenerFHKMPW01}
if there are symmetries $\{ \symm_{i_1, i_2} \}_{i_1, i_2 \in [1, n]} \subseteq \Aut(\F)$
that swap rows $i_1$ and $i_2$ via:
\begin{equation*}
    \symm_{i_1, i_2} (\ell) =
    \begin{cases}
        \ell_{i_2, j} & \text{if } \ell = \ell_{i_1, j} \text{ for some } j \in [1, m] \\
        \ell_{i_1, j} & \text{if } \ell = \ell_{i_2, j} \text{ for some } j \in [1, m] \\
        \ell_{i, j} & \text{if } \ell = \ell_{i, j} \in M \text{ and } i \not\in \{i_1, i_2\}\\
    \end{cases}.
\end{equation*}
Note that for our purposes,
row swaps are free to affect literals~$\Lit(\F) \setminus M$ that lie outside of the matrix.
By composing row swaps,
every possible reordering of the rows can be achieved.
In group-theoretic terms,
these swaps generate the symmetric group over the rows.
Row symmetry and related structures 
are crucial for practical symmetry-handling algorithms,
and they can be detected by
generator-based~\cite{DevriendtBBD16,PfetschR19,HojnyP19}
or more-recently developed graph-based approaches~\cite{AndersBR24}.

We now turn to orbitopal fixing.
Suppose our formula~$F$ has a matrix of literals~$M$
that exhibits row symmetry.
If each column of $M$ is a unique literal clause of $F$,
then we may fix
the bottom literal in the first column~$\ell_{n, 1}$ to true
and all literals in the upper-triangular portion above $\ell_{n, 1}$ to false.
Figure~\ref{fig:orbitopal} shows an example of orbitopal fixing.
More formally:

\begin{definition}[Orbitopal fixing]
    Let $F$ be a formula,
    and let $M := (\ell_{i, j})$ be an $n \times m$ matrix
    that exhibits row symmetry in $F$.
    If every column of $M$ is ULC with respect to $F$,
    i.e.,
    if $C_j \coloneqq \{ \ell_{i, j} \}_{i \in [1, n]} \in F$
    is ULC for every $j$,
    then \emph{orbitopal fixing} derives unit clauses
    $(\ell_{n, 1})$ and $(\ol{\ell_{i, j}})$ 
    for every $j \in [1, \min(n, m)]$
    and $i \leq n - j$.\label{def:orbitopal}
\end{definition}

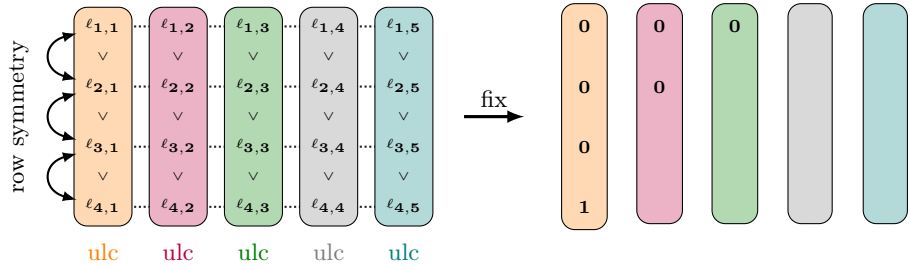
\begin{figure}[t]
    \centering
    \begin{tikzpicture}[scale=0.8] 
        \begin{scope}
            \foreach \x in {0,...,4}
            \foreach \y in {0,...,3}
            {
                \pgfmathtruncatemacro{\xx}{\x+1}%
                \pgfmathtruncatemacro{\yy}{\y+1}%
                \node []  at (1.25*\x,-1*\y) (l\y\x) {\tiny $\mathbf{\ell_{\yy,\xx}}$};
            } 

            \foreach \x in {0,...,4}
            \foreach \y in {0,...,2}
            {
                \node []  at (1.25*\x,-1*\y-0.5) (v\y\x) {\tiny $\vee$};
            } 
            \foreach \x in {0,...,3}
            \foreach \y in {0,...,3}
            {
                \pgfmathtruncatemacro{\n}{\x + 1}%
                \begin{pgfonlayer}{background}
                    \draw [thick,densely dotted,gray!30!black] (l\y\x) to [] (l\y\n);
                \end{pgfonlayer}
            } 
            \node [rotate=90] at (-1.4,-1.5) (ulc3) {\footnotesize row symmetry};

            \begin{pgfonlayer}{background}
                \node[fit=(l00)(l10)(l20)(l30), fill=orange!30, draw=black, rounded corners=5,inner sep=1] {};
                \node[fit=(l01)(l11)(l21)(l31), fill=purple!30, draw=black, rounded corners=5,inner sep=1] {};
                \node[fit=(l02)(l12)(l22)(l32), fill=green!50!black!30,   draw=black, rounded corners=5,inner sep=1] {};
                \node[fit=(l03)(l13)(l23)(l33), fill=gray!30,   draw=black, rounded corners=5,inner sep=1] {};
                \node[fit=(l04)(l14)(l24)(l34), fill=teal!30,  draw=black, rounded corners=5,inner sep=1] {};
            \end{pgfonlayer}
                \node [color=orange] at (1.25*0,-3.75) (ulc0) {\footnotesize ulc};
                \node [color=purple] at (1.25*1,-3.75) (ulc1) {\footnotesize ulc};
                \node [color=green!50!black] at (1.25*2,-3.75) (ulc2) {\footnotesize ulc};
                \node [color=gray] at (1.25*3,-3.75) (ulc3) {\footnotesize ulc};
                \node [color=teal] at (1.25*4,-3.75) (ulc4) {\footnotesize ulc};

                \draw [latex-latex,thick,] (l00) to[bend right=75,looseness=2] (l10);
                \draw [latex-latex,thick,] (l10) to[bend right=75,looseness=2] (l20);
                \draw [latex-latex,thick,] (l20) to[bend right=75,looseness=2] (l30);
        \end{scope}

        \draw [-latex,very thick] (6,-1.5) -- node[midway,above] {fix} (7,-1.5);
        
        \begin{scope}[xshift=8cm] 
            \def\fixlist{0,0,0,\phantom{0},\phantom{0},0,0,,,,0,,,,,1,,,,}
            \foreach \v [count=\i from 0] in \fixlist {%
                \pgfmathtruncatemacro{\x}{mod(\i,5)}%
                \pgfmathtruncatemacro{\y}{floor(\i/5)}%
                \node []  at (1.25*\x,-1*\y) (l\y\x) {\scriptsize{$\mathbf{\v}$}};
            } 

            \begin{pgfonlayer}{background}
                \node[fit=(l00)(l10)(l20)(l30), fill=orange!30, draw=black, rounded corners=5] {};
                \node[fit=(l01)(l11)(l21)(l31), fill=purple!30, draw=black, rounded corners=5] {};
                \node[fit=(l02)(l12)(l22)(l32), fill=green!50!black!30,   draw=black, rounded corners=5] {};
                \node[fit=(l03)(l13)(l23)(l33), fill=gray!30,   draw=black, rounded corners=5] {};
                \node[fit=(l04)(l14)(l24)(l34), fill=teal!30,  draw=black, rounded corners=5] {};
            \end{pgfonlayer}
        \end{scope}
    \end{tikzpicture}
    \caption{
        An example of orbitopal fixing applied to a $4 \times 5$ matrix of literals from some formula~$F$.
        The matrix exhibits row symmetry,
        and the columns are ULCs of $F$.
        Putting these two conditions together,
        we may fix the bottom-left literal~$\ell_{4, 1}$ to true
        and the upper-triangular portion of literals above $\ell_{4, 1}$ to false.
    } \label{fig:orbitopal}
\end{figure}

At first,
it might be surprising that we may fix so many literals at once.
But by using the property of ULCs from Lemma~\ref{lem:ulc},
we may assume that every column of $M$ is satisfied by exactly one literal.
This assumption allows us to strategically
swap the rows with the satisfied literals
into the lower-triangular portion of $M$,
thus allowing us to fix the upper-triangular portion to false.
This argument is formalized in the following lemma.

\begin{theorem}
    Let $\F$ be a formula with the conditions from Definition~\ref{def:orbitopal},
    and let $\F \land L$ be the formula obtained by applying the orbitopal fixing rule to $\F$.
    Then the following hold:
    \begin{enumerate}
        \item $\F$ and $\F \land L$ are equisatisfiable.
        \item There exists an \sr{} proof that adds the unit clauses of $L$ to $\F$ in a particular order,
              namely, column-wise, top to bottom, left to right.
    \end{enumerate}
\end{theorem}

\begin{proof}
    \textbf{(1).}
    By Lemma~\ref{lem:ulc},
    let $\assn$ be a satisfying assignment for $\F$
    that satisfies each column~$C_j$ of the matrix~$M$
    with exactly one true literal.
    We will now transform $\assn$ into a new assignment
    that also satisfies the additional constraints in $\F \land L$.

    First, consider the leftmost column.
    Let $i$ be the row containing the satisfied literal of $C_1$ under $\tau$.
    If $i \neq n$,
    then we may use the row symmetry~$\sigma_{i, n}$
    to change $\assn$
    into a new assignment $\assn \circ \sigma_{i, n}$
    by swapping the truth values for rows $i$ and $n$.
    Since $\symm_{i, n}$ is a symmetry of $\F$,
    the assignment $\assn \circ \sigma_{i, n}$ still satisfies $\F$,
    and by our assumption that $\assn$ satisfies exactly one literal of $C_1$,
    it also satisfies the unit clauses in $L$
    setting $\ell_{n, 1}$ to true and $\ell_{i, 1}$ to false for all $i < n$.

    Now consider the $j$-th column,
    and let $i$ be the row containing the satisfied literal of $C_j$.
    If $i \leq n - j$,
    then we may do the same thing as before
    and swap rows $i$ and $n - j + 1$
    to form $\assn \circ \symm_{i, n - j + 1}$.
    Note that swapping these two rows exchanges only falsified literals
    in the columns to the left,
    since their true literals lie beneath the $(n - j + 1)$-th row.
    As a result,
    we still satisfy the constraints in these columns.
    And since the true literal of $C_j$ is now beneath row $n - j$,
    the new assignment sets every literal at and above this row to false,
    which satisfies the unit clauses in $L$ corresponding to column $j$.
    
    By applying this procedure to the columns in order from left to right,
    we obtain a modified $\assn$ that satisfies both $F$ and $F \land L$.\smallskip

    \textbf{(2).} The series of \sr{} clause additions
    generated by orbitopal fixing
    follows the proof of \textbf{(1)} exactly,
    except we must add the clauses 
    in order (top to bottom, left to right),
    and we must provide a row-permutation witness for each unit clause.
    More specifically,
    when we add the unit clause~$\{ \ol{\ell_{i, j}}\}$,
    we use the witness~$\wit$ with 
    $$\wit(\ell) \coloneqq \begin{cases} 
        \bot & \text{ if } \ell = \ell_{i,j} \\
        \top & \text{ if } \ell = \ell_{i+1,j} \\
        \sigma_{i, i+1}(\ell) & \text{ otherwise.}\\
    \end{cases}$$
    Let $\F \land L'$ be the formula constructed so far by adding unit clauses.

    According to Definition~\ref{def:sr},
    for us to show that $\{ \ol{\ell_{i, j}}\}$ is \sr,
    we must show that
    $F \land L' \land \{ \ell_{i, j} \} \vdash_1 (F \land L' \land \{ \ol{\ell_{i, j}} \})_{| \wit}$.
    For most clauses,
    this is trivial:
    any clauses $C \in F$ whose literals are not modified by $\wit$,
    i.e.,
    where $C \cap \{ \ell \in \Lit(\F) \;|\; \ell \neq \symm(\ell) \} = \varnothing$,
    are immediately entailed.
    The new unit clause~$\{ \ol{\ell}_{i, j} \}$ is also trivially satisfied,
    since $\wit(\ell_{i, j}) = \bot$.

    The remaining types of clauses modified by the witness are:
    (i)~previously added unit literals $L'$,
    (ii)~clauses of $F$ \emph{not} containing the variables of~$\ell_{i,j}$ and~$\ell_{i+1,j}$, and
    (iii)~clauses of~$F$ containing the variables of~$\ell_{i,j}$ and~$\ell_{i+1,j}$.

    \emph{(Case i.)}
    The order of the unit additions
    ensures that all previously added unit literals~$\ell_{i,j'}$
    for~$1 \leq j' < j$ are 
    mapped to other propagated literals~$\ell_{i+1,j'}$.
    Hence, they also entail each other.

    \emph{(Case ii.)} 
    For every clause $C \in F$
    not containing the variables of $\ell_{i, j}$ or $\ell_{i + 1, j}$,
    we observe that the symmetry $\sigma_{i,i+1}$ applies,
    and $\wit(C) \in F$ holds. 

    \emph{(Case iii.)}
    The ULC containing $\ell_{i,j}$ and $\ell_{i+1,j}$ is entailed,
    since $\wit (\ell_{i+1,j}) = \top$.
    And because this clause is a ULC,
    all other clauses may only contain 
    $\ol{\ell_{i,j}}$ and $\ol{\ell_{i+1,j}}$.
    All clauses containing $\ol{\ell_{i,j}}$ are immediately entailed.

    It remains to show that clauses~$C$ containing $\ol{\ell_{i+1,j}}$ but not $\ol{\ell_{i,j}}$ are entailed.
    Consider the clause~$C'$ 
    we obtain by mapping $C$ 
    under the row swap exchanging rows $i$ and $i+1$,
    that is,
    $C' = \sigma_{i,i+1}(C)$.
    Since the row swap is a symmetry of $F$, $C' \in F$ holds.
    In other words, $C'$ is a premise.

    By assumption, $C \cap \{\ol{\ell_{i,j}},\ell_{i,j},\ell_{i+1,j}\} = \varnothing$ and $\ol{\ell_{i+1,j}} \in C$ hold, and thus we can conclude
    $C' = C_{| \wit} \cup \{\ol{\ell_{i,j}}\}$.
    Assuming~$\neg C_{| \wit}$ together with the additional premise~$\ell_{i,j}$
    thus contradicts the premise $C'$.
    Hence, $C_{| \wit}$ is entailed.
    \hfill$\square$

\end{proof}

\begin{example}
    The pigeonhole problem with $m$ pigeons and $n$ holes
    exhibits row symmetry
    when encoded as in Example~\ref{ex:sr-sec-php},
    with the rows corresponding to the holes and the columns corresponding to the pigeons.
    Thus, orbitopal fixing may be applied.
    Figure~\ref{fig:orbitopal} illustrates the case of 5 pigeons and 4 holes.
    When applying orbitopal fixing to this formula,
    we end up with the following formula:
    \begin{equation*}
        \textsc{php} (5, 4)
        \land \ol{p}_{1, 1} \land \ol{p}_{2, 1} \land \ol{p}_{3, 1} \land p_{4, 1}
        \land \ol{p}_{1, 2} \land \ol{p}_{2, 2}
        \land \ol{p}_{3, 1}.
    \end{equation*}
\end{example}

\subsection{Clausal Fixing}

Our second symmetry-breaking technique, called \emph{clausal fixing}, 
is based on the observation that every clause must contain at least one satisfied literal. 
When all literals of a clause belong to the same orbit under the formula's symmetries, 
a representative literal can be fixed without loss of generality.

\begin{definition}[Clausal fixing]
Let $\F$ be a formula with clause~$\{\ell_1, \dots{}, \ell_k\} \in \F$
where all $\ell_i$ are in the same orbit of $\Aut(\F)$. 
That is, 
there are symmetries $\symm$ with $\symm(\ell_1) = \ell_i$ for each $i$. 
Then \emph{clausal fixing} derives the unit clause $\{\ell_1\}$.
\end{definition}

Intuitively, if we have any satisfying assignment, 
we know that there must be at least one satisfied literal in that clause. 
Using the orbit, 
we can always swap the satisfied literal with $\ell_1$.
Thus, we can just assign $\ell_1$ directly.

We now formally prove the correctness of the clausal fixing rule.
\begin{theorem} 
    Let $\F \land \{\ell_1\}$
    be the formula obtained from $\F$
    by an application of the clausal fixing rule
    to clause~$C = \{\ell_1, \dots{}, \ell_k\}$.
    Then the following hold:
    \begin{enumerate}
        \item $\F$ and $\F \land \{ \ell_1 \}$ are equisatisfiable.
        \item There is an \sr{} proof deriving $F \land \{\ell_1\}$ from $F$ in $k$ steps.
    \end{enumerate}
\end{theorem}
\begin{proof}
    \textbf{(1).}
    Let $\assn$ be a satisfying assignment of $\F$. 
    We show how to transform $\assn$ 
    into an assignment that satisfies $\F \land \{\ell_1\}$.

    Let $\ell_i$ be a satisfied literal of $C$.
    If $\ell_i = \ell_1$,
    then $\assn$ would already satisfy $F \land \{\ell_1\}$,
    so assume otherwise.
    By definition of the clausal fixing rule, 
    for every literal $\ell_i \in C$,
    there exists a symmetry~$\sigma$ with $\sigma(\ell_1) = \ell_i$.
    Using this symmetry,
    we obtain $\assn' = \assn \circ \sigma$, 
    which now sets $\ell_1$ to true.
    Since $\sigma$ is a symmetry, 
    the resulting assignment $\assn'$ still satisfies $F$.

    \textbf{(2).}
    We obtain an \sr{} proof as follows. 
    First,
    the proof derives binary symmetry-breaking clauses
    $\{\ell_1, \ol{\ell}_i\}$ for all $i \in [2,k]$.
    Each of these clauses is \sr{} using the symmetry $\wit_i$ mapping $\wit_i(\ell_1) = \ell_i$.
    After that,
    we may derive the unit clause~$\{ \ell_1 \}$
    by resolution on the added binary clauses.

    It suffices to show that each binary clause $\{ \ell_1, \ol{\ell}_i \}$ is \sr.
    Let $L'$ be the set of binary clauses we have already added.
    We must show that
    \begin{equation*}
        \F \wedge L' \wedge \{ \ol{\ell}_1 \} \wedge \{ \ell_i \} \vdash_1 (\F \wedge L' \wedge \{ \ell_1, \ol{\ell}_1 \})_{| \wit_i}.
    \end{equation*}
    The result is immediate:
    Every clause in $\F$ is entailed,
    since $\wit_i$ is a symmetry of $\F$,
    and every clause in $C \in L' \cup \{ \ell_1, \ol{\ell}_1 \}$
    is entailed by the unit clause $\{ \ell_i \}$,
    since $\ell_1 \in C$
    and $\wit_i (\ell_1) = \ell_i$. \hfill$\square$
\end{proof}


\subsection{Negation Fixing}
Lastly, we describe the \emph{negation fixing} rule.
When a literal can be mapped to its negation,
we can fix it without loss of generality.
\begin{definition}[Negation fixing]
Let $F$ be a formula,
let $\ell \in \Lit(F)$ be a literal,
and let $\sigma \in \Aut(F)$
be a symmetry that maps $\sigma(\ell) = \ol{\ell}$.
Then \emph{negation fixing} derives the unit clause $\{\ell\}$. 
\end{definition}

In a sense, the negation fixing rule also exploits cardinality: trivially, at least one of~$\ell$ and~$\ol \ell$ must be true. 
We prove the correctness of the rule.

\begin{theorem} 
    Let $\F \land \{\ell\}$ be a formula obtained from $\F$ by an application of the negation fixing rule.
    Then the following hold: 
    \begin{enumerate}
        \item $\F$ and $\F \land \{\ell\}$ are equisatisfiable.
        \item The unit clause $\{\ell\}$ is \sr{}.
    \end{enumerate}
\end{theorem}
\begin{proof}
    \textbf{(1.)}
    Let $\assn$ be a satisfying assignment of $\F$. 
    If $\assn$ sets $\ell$ to true,
    then we're done,
    so assume otherwise.
    Then the assignment $\assn \circ \symm$ satisfies $\F \wedge \{ \ell \}$,
    where $\symm$ is the symmetry from the negation fixing rule swapping $\ell$ with $\ol{\ell}$.

    \textbf{(2.)} The clause $\{\ell\}$ is \sr{} using as witness the symmetry mapping $\sigma(\ell) = \ol \ell$.
    \hfill$\square$
\end{proof}


\subsection{Repeated Applications of Rules}\label{subsec:repeated-rules}

All of our rules use the symmetries of a formula $\F$ 
to add a set of unit clauses~$L$ to the formula,
yielding a new formula $\F \land L$.
In turn, subsequent rule applications would use the symmetries of $F \land L$. 
However, recomputing formula symmetries after \emph{every} 
rule application is not practical.

Instead of recomputing symmetries,
we can \emph{update} the set of applicable symmetries 
using pointwise and setwise stabilizers (see Section~\ref{sec:symmetries}).
This approach may indeed yield fewer applicable symmetries than 
computing the full group~$\Aut(F \land L)$,
but it's cheaper to do so.

Let us now observe that by stabilizing the set of added unit literals $L$, we obtain symmetries of $F \land L$ 
from the symmetries of $F$.
\begin{lemma} 
    Let $F$ be a CNF formula and $L \subseteq \Lit(F)$ a subset of its literals.
    It holds that $\Aut(F)_{\{L\}} \subseteq \Aut(F \land \{\{\ell\} \;|\; \ell \in L\})$.
\end{lemma}
\begin{proof}
    Let $\sigma \in \Aut(F)_{\{L\}}$.
    This means that $\sigma(F) = F$ and $\sigma(L) = L$.
    Consider a clause $C$ of the formula $F \land \{\{\ell\} \;|\; \ell \in L\}$.
    If $C \in F$, then $\sigma(C) \in F \land L$ follows.
    If $C = \{\ell\}$ is one of the unit clauses from $L$, then $\sigma(\ell) \in L$ follows, 
    and $\{\sigma(\ell)\} \in \{\{\ell\} \;|\; \ell \in L\}$.
    Hence, $\sigma$ is a symmetry of $F \land \{\{\ell\} \;|\; \ell \in L\}$.
    \hfill$\square$
\end{proof}

Unfortunately,
setwise stabilizers are also expensive to compute. 
(In fact, the problem of computing them 
is at least as hard as computing symmetries~\cite{Luks91}.)
However, since $\Aut(F)_{(L)} \subseteq \Aut(F)_{\{L\}}$ holds, 
we can use pointwise stabilizers instead.
Polynomial-time algorithms exist to find pointwise stabilizers~\cite{seress_2003},
and they are often efficient in practice.

\section{Implementation Details}
\label{sec:implementation}

We implemented our new fixing algorithms in the existing 
\satsuma{} symmetry breaking tool~\cite{AndersBR24}.
The tool is implemented in \texttt{C++}.

To enable efficient symmetry handling,
\satsuma{} simplifies the formula in various ways:
Duplicate literals are removed from clauses,
duplicate clauses and tautological clauses are removed from the formula,
and unit propagation is applied until fixpoint.
But other than these simplifications,
the tool only adds symmetry-breaking clauses.
In particular, the ordering of literals and clauses in the formula remains unchanged.

\paragraph{Orbitopal Fixing.} 
For the orbitopal fixing approach,
we leverage the existing structure detection of \satsuma{} for row symmetry,
row-column symmetry,
and the symmetries of so-called Johnson graphs~\cite{AndersBR24}.\footnote{
    A Johnson graph $J(n, k)$ represents the $k$-element subsets of an $n$-element set,
    where two vertices share an edge if their set intersection has cardinality $(k - 1)$.
    Intuitively,
    Johnson graphs model the symmetries of the edges of complete graphs,
    or,
    more generally,
    relational structures.
}
Although orbitopal fixing is only defined over row symmetry, 
the more complex structures identified
by \satsuma{} often \emph{contain} a row symmetry.
In particular,
row-column symmetry is inherently composed of two row symmetries.
For Johnson symmetry,
the relationship is more intricate, 
but certain instances of this structure can also contain an underlying row symmetry.

Once a structure is identified as potentially having row symmetry,
we check its columns to see if they coincide with a unique literal clause in the formula.
Orbitopal fixing is applied to columns that fulfill the condition.

\paragraph{Clausal Fixing.}
Our implementation of clausal fixing follows a four-step procedure:
\begin{enumerate}
    \item \textit{(Orbits.)} 
         We first compute the orbits of the currently considered group. 

    \item \textit{(Check clauses.)} 
          Each clause is considered once. 
          If all of its literals belong to the same orbit,
          then the clausal fixing rule is applied
          and one literal is propagated.
          Since orbit partitions refine strictly under pointwise stabilizers, 
          a clause that does not qualify at this stage will not qualify at any later stage.

      \item \textit{(Witness.)} 
          To generate \sr{} proofs, 
          we must identify symmetries that can act as witnesses.
          When we apply clausal fixing to a literal~$\ell$ in a clause~$C$, 
          we explicitly compute 
          symmetries $\symm$ such that $\symm(\ell) = \ell'$ for each $\ell' \in C$.

      \item \textit{(Stabilize.)} Whenever propagation occurs, 
            we update the group by taking the pointwise stabilizer 
            of the propagated literals.
            Then we go back to Step~1
            and recompute the orbits for the refined group.
\end{enumerate}
The process stops once each clause has been checked once.

We employ two different algorithms to compute pointwise stabilizers and witness symmetries:
a more involved Schreier-Sims-based implementation,
and a faster heuristic for binary clauses.
\begin{enumerate}
    \item The Schreier-Sims algorithm \cite{SIMS1970169,seress_2003} 
        computes pointwise stabilizers.
        Internally, it stores a so-called transversal which can immediately 
        provide the necessary witness symmetries.
        While it tends to be quite fast for many groups,
        it often exhibits quadratic scaling in practice. 
        As a result,
        we set computational limits on the use of Schreier-Sims in our implementation.
        We use the Schreier-Sims implementation of the \dejavu{}~\cite{AndersS21} library.

    \item For large groups, we implement a more rudimentary heuristic 
        which only tests binary clauses.
        It greedily searches for
        an existing generator that can serve as the witness symmetry.
        The heuristic takes pointwise stabilizers by 
        filtering generators to ones that stabilize the desired points.
\end{enumerate}

\paragraph{Negation Fixing.}
Negation fixing follows a very similar strategy to clausal fixing,
except instead of iterating over clauses,
it iterates over variables. 
For each variable~$v$, 
we check if $v$ and $\ol v$ are in the same orbit.
As with clausal fixing, 
we employ both a Schreier-Sims based implementation 
and a more efficient heuristic.


\begin{figure}[t]
\centering
\includegraphics[width=\textwidth]{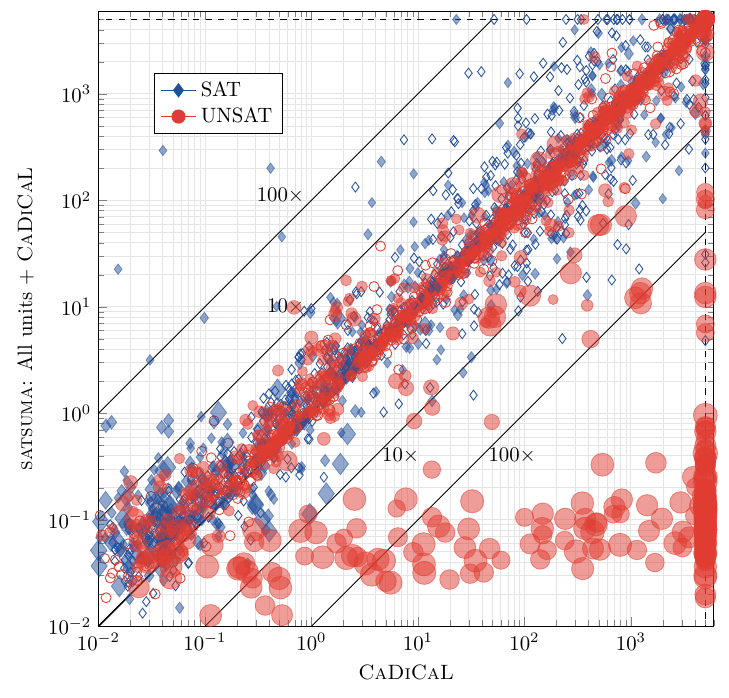}
\vspace{-20pt}
\caption{A scatter plot of \cadical{} with and without fixing on the anniversary suite.
The times are in seconds and include preprocessing time.
Empty marks denote that no units were added.
The size of the marks correlates to the number of fixed units relative to the number of formula variables.
Points below the diagonal benefited from fixing.}
\label{fig:scatter-anni}
\end{figure}

\section{Experimental Evaluation}

We evaluated the effectiveness of our fixing techniques as implemented in \satsuma{} on three benchmark suites.
The first suite is from
\href{https://benchmark-database.de/?query=track%3Danni_2022}{the anniversary track}
of the SAT Competition 2022~\cite{satcomp2022},
comprising non-random benchmarks from the SAT Competitions 2002 to 2021,
duplicates excluded.
The second suite comprises the benchmarks from the
\href{https://benchmark-database.de/?query=track%3Dmain_2025}{main track}
of the SAT Competition 2025.
The third suite
comprises highly symmetric synthetic benchmarks
that have appeared in various papers on symmetry breaking,
including pigeonhole, Tseitin, clique coloring, and graph coloring formulas,
as well as multiple instances from Ramsey theory.   

For each suite, we filtered out all formulas larger than 1 GB in size, 
since for such large formulas symmetry handling often incurs 
prohibitive computational overhead,
and are thus skipped by the symmetry breaking tools anyway.
After this filter,
the anniversary suite has 5344 formulas,
the 2025 suite has 386 formulas,
and the synthetic suite has 137 formulas.

We ran our experiments on the cluster at the \href{https://www.psc.edu/}{Pittsburgh Supercomputing Center}~\cite{bridges2}.
Each machine has 128 cores and 256 GB of RAM.
We ran every tool in parallel across all cores.
We used a timeout of 5000 seconds for \cadical{}\footnote{
    \url{https://github.com/arminbiere/cadical}. We used version 2.1.3.
}
(which matches the official timeout for the SAT competitions),
as well as a 300 second timeout for \satsuma{} and \dsrtrim~\cite{CodelAH24},\footnote{
    \url{https://github.com/ccodel/dsr-trim}.
}
an (unverified) \sr{} proof checker.
Notably,
all \satsuma{} runs finished before timeout.

We ran \satsuma{}
with five different settings:
each of our three fixing techniques individually,
all of our techniques combined (``all-units''),
and a control version of \satsuma{} that produces lex-leader constraints.
We checked all \satsuma-generated \sr{} proofs with \dsrtrim{}.
Then
we ran \cadical{} on all formulas.

Figure~\ref{fig:scatter-anni} shows the results
of applying all fixing techniques to the anniversary benchmarks.
Our fixing techniques allow \cadical{} to solve many dozens of formulas very quickly,
including 67 formulas that can be solved in a second of preprocessing time,
while \cadical{} without symmetry breaking times out after 5000 seconds.
Many empty points below the diagonal
are due to \satsuma{}'s formula simplification
(especially on satisfiable instances).
Note that few points are clearly above the diagonal,
thereby showing that our techniques have minimal negative impact on the performance
across the anniversary benchmarks. 

Figure~\ref{fig:scatter-synth-sc25} shows similar results on the synthetic and 2025 suites.
Notably,
most of the synthetic benchmarks become easy after fixing,
showing, somewhat surprisingly,
that symmetry breaking is all that is needed to solve these instances. 

\begin{figure}[h]
\centering
\includegraphics[width=.49\textwidth]{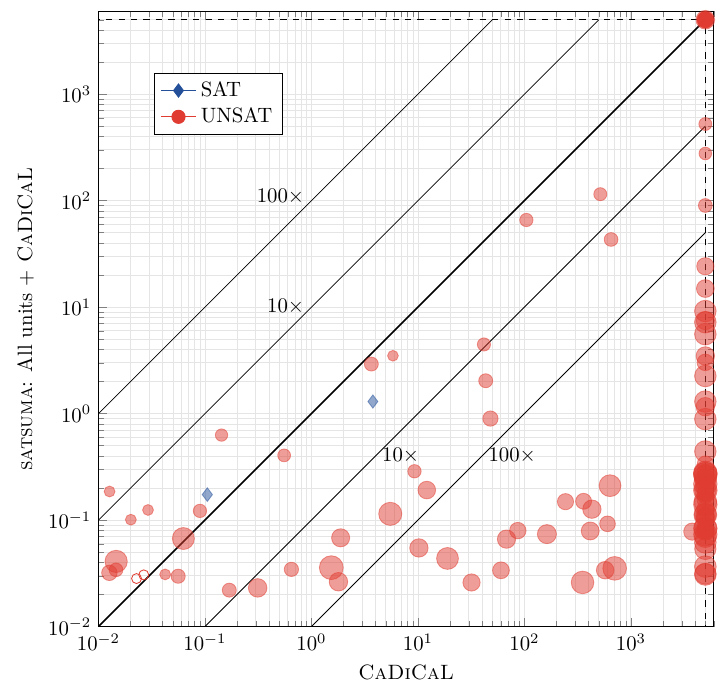}
\hfill
\includegraphics[width=.49\textwidth]{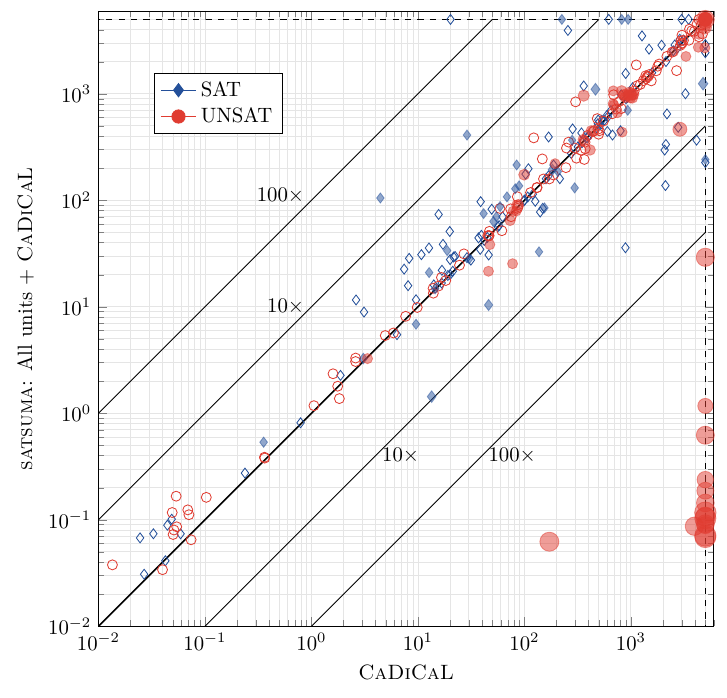}
\vspace{-10pt}
\caption{Scatter plots of \cadical{} with and without fixing on the synthetic suite (left) and the 2025 suite (right).
The times are in seconds and include preprocessing time.}
\label{fig:scatter-synth-sc25}
\end{figure}

Figure~\ref{fig:cactus-anni} shows a cumulative solved formulas (CSF) plot
comparing the runtimes of \cadical{} with each \satsuma{} setting against base \cadical{} on the anniversary suite.
The left CSF plot shows the regression due to symmetry breaking on satisfiable instances,
since symmetry breaking on satisfiable formulas typically has no benefit and can be harmful instead.
The plot shows that the lex-leader setting performs the worst,
while the new techniques limit the amount of regression.
The regression is similar across \satsuma{} settings,
which suggests that \satsuma{}'s formula preprocessing and simplification dominate the costs.

The right CSF plot of Figure~\ref{fig:cactus-anni} shows the runtime performance on unsatisfiable instances.
Lex-leader performs the best,
but the all-units setting is close behind.
Below them,
each individual setting performs similarly,
with clausal fixing performing the worst.
All techniques outperform base \cadical.

\begin{figure}[h]
\centering
\includegraphics[width=.49\textwidth]{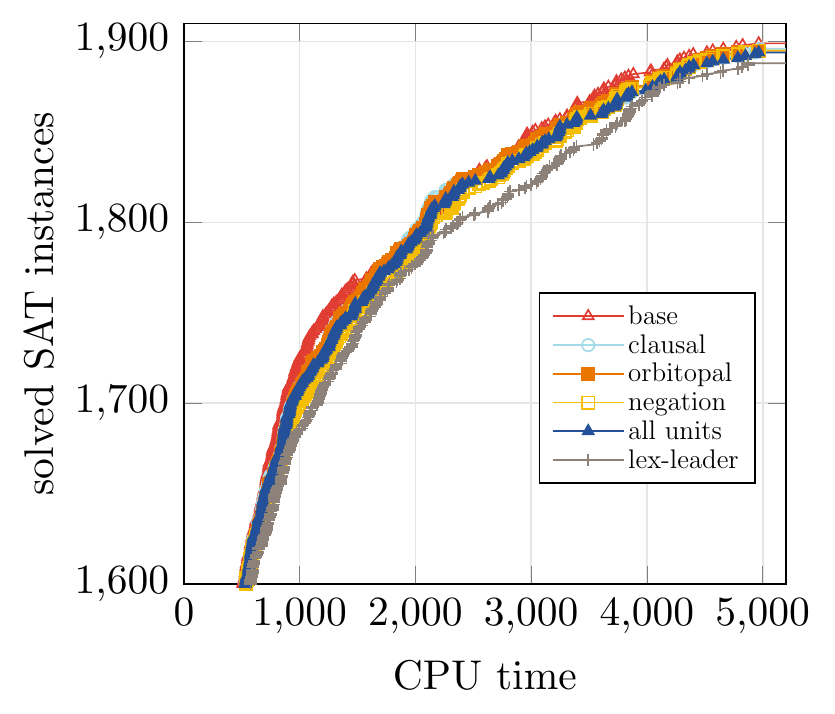}
\hfill
\includegraphics[width=.49\textwidth]{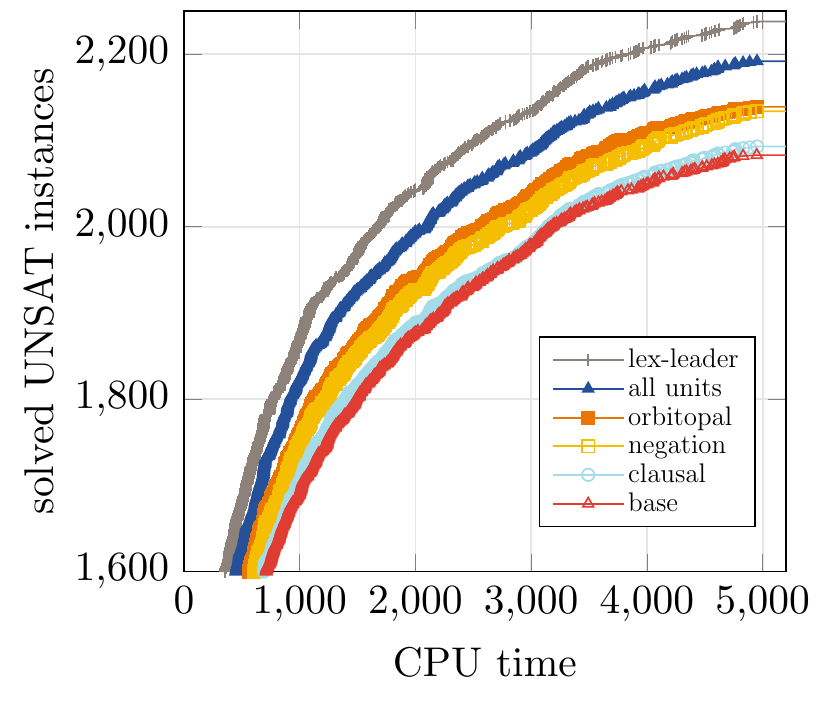}\vspace{-10pt}
\caption{CSF plots on SAT (left) and UNSAT (right) formulas of the anniversary suite}
\label{fig:cactus-anni}
\end{figure}

Figure~\ref{fig:cactus-synth-sc25} shows the CSF plots for the synthetic and 2025 suites. 
On the synthetic instances, the all-units configuration performs the best.
The main reason for its success is that negation fixing can easily solve all Tseitin formulas,
in contrast to lex-leader.
(A similar approach could be implemented for lex-leader,
but this is currently not present in \satsuma{}.)
On the 2025 suite,
lex-leader shows the strongest performance, 
followed by the all-units configuration.

\begin{figure}[t]
\centering
\includegraphics[width=.49\textwidth]{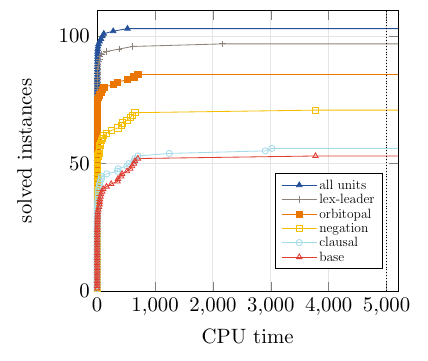}
\hfill
\includegraphics[width=.49\textwidth]{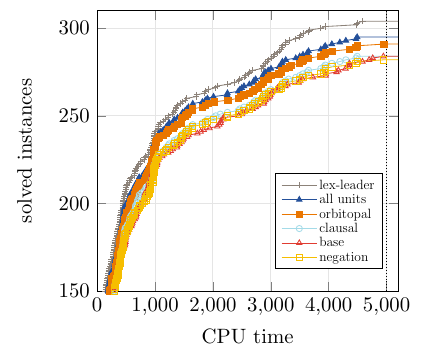}
\vspace{-10pt}
\caption{CSF plots on the synthetic suite (left) and the 2025 suite (right)}
\label{fig:cactus-synth-sc25}
\end{figure}

\begin{table}[t!]
    \caption{
        Results across the benchmark suites for the average preprocessing time in seconds (PPT),
        the average \cadical{} runtime with preprocessing in seconds (CRT),
        and the average number of units added to the formula by the fixing technique (\#U).
    }
    \label{fig:table:summary-stats}
    \centering
    \begin{tabular}{c@{~\,~}r@{~\,~}r@{~\,~}r@{~\,~}r@{~\,~}r@{~\,~}r@{~\,~}r@{~\,~}r@{~\,~}r}
    \toprule
        \multicolumn{1}{c}{\makebox[0.5in]{Suite}}
        & \multicolumn{3}{c}{\makebox[0.33in]{Anni22}}
        & \multicolumn{3}{c}{\makebox[0.33in]{Satcomp25}}
        & \multicolumn{3}{c}{\makebox[0.33in]{Synth}} \\
        \cmidrule(l){2-4}\cmidrule(l){5-7}\cmidrule(l){8-10}
        \multicolumn{1}{c}{\makebox[0.5in]{Setting}}
        & \makebox[0.33in]{PPT} & \makebox[0.33in]{CRT} & \makebox[0.33in]{\#U}
        & \makebox[0.33in]{PPT} & \makebox[0.33in]{CRT} & \makebox[0.33in]{\#U}
        & \makebox[0.33in]{PPT} & \makebox[0.33in]{CRT} & \makebox[0.33in]{\#U} \\
        \midrule
        Orbitopal   & 2.51 & 1555.5 & 11.23    &    8.06 & 1746.7  & 8.37   &     0.44 & 1919.2 & 743.0  \\
        Negation    & 1.93 & 1572.0 & 32.59    &    5.77 & 1870.6  & 52.15  &     0.32 & 2469.4 & 16.5   \\
        Clausal     & 1.83 & 1605.7 & 23.64    &    5.45 & 1845.3  & 14.33  &     0.73 & 3041.0 & 15.2   \\
        All         & 1.81 & 1504.5 & 67.30    &    5.27 & 1694.3  & 71.48  &     0.76 & 1250.0 & 768.6  \\
        Lex-leader  & 1.73 & 1472.4 & --       &    4.85 & 1593.7  & --     &     0.32 & 1486.6 & --     \\
        \midrule
        \cadical{}  & --   & 1607.8 & --       &    --   & 1882.9  & --     &     --   & 3138.6 & --     \\
    \bottomrule
    \end{tabular}
\end{table}

\begin{table}[t!]
    \caption{
        Results for \sr{} proof checking times across the various benchmark suites and settings, showing the 
        average \dsrtrim{} checking time in seconds (DCT) and
        the average \lsrcheck{} checking time in seconds (LCT).
    }
    \label{fig:table:proof-stats}
    \centering
    \begin{tabular}{c@{~\,~}c@{~\,~}c@{~\,~}c@{~\,~}c@{~\,~}c@{~\,~}c}
    \toprule
        \multicolumn{1}{c}{\makebox[1.1in]{Suite}}
        & \multicolumn{2}{c}{\makebox[1.1in]{Anni22}}
        & \multicolumn{2}{c}{\makebox[1.1in]{Satcomp25}}
        & \multicolumn{2}{c}{\makebox[1.1in]{Synth}} \\
        \cmidrule(l){2-3}\cmidrule(l){4-5}\cmidrule(l){6-7}
        \multicolumn{1}{c}{\makebox[0.7in]{Setting}}
        & \multicolumn{1}{c}{\makebox[0.5in]{DCT}} & \multicolumn{1}{c}{\makebox[0.5in]{LCT}}
        & \multicolumn{1}{c}{\makebox[0.5in]{DCT}} & \multicolumn{1}{c}{\makebox[0.5in]{LCT}}
        & \multicolumn{1}{c}{\makebox[0.5in]{DCT}} & \multicolumn{1}{c}{\makebox[0.5in]{LCT}} \\
        \midrule
        Orbitopal   & 5.35 & 0.76 &    15.15 & 1.36  &     8.42 &  4.19 \\
        Negation    & 5.76 & 1.43 &    16.80 & 1.64  &     0.04 &  0.01 \\
        Clausal     & 6.76 & 0.86 &    16.13 & 1.29  &     5.29 &  0.16 \\
        All         & 7.48 & 1.41 &    15.91 & 1.48  &     8.67 &  3.18 \\
    \bottomrule
    \end{tabular}
\end{table}

Table~\ref{fig:table:summary-stats} summarizes the data shown in the figures.

As part of our experiments,
we generated and checked \sr{} and \drat{} proofs.
\satsuma{} generated \sr{} proofs when applying our fixing techniques,
and \cadical{} generated \drat{} proofs for unsatisfiable formulas.
All proofs were either accepted by the \dsrtrim{} and \lsrcheck{} proof checkers,
or caused a memout/timeout.
Proof checking time was generally low,
with most proofs finishing in under 15 seconds.
Table~\ref{fig:table:proof-stats} summarizes the times taken for proof checking.

One particular advantage of \satsuma's \sr{} proof generation is that it composes with \cadical's \drat{} proof generation.
For any unsatisfiable symmetry-broken formula,
we appended the \drat{} proof to the end of its \sr{} proof,
and then checked the proof with respect to the original formula.
In this way,
the two proofs form a complete proof of unsatisfiability of the original formula.

\section{Conclusion}
We presented new static symmetry-breaking techniques 
based solely on introducing unit clauses to a given formula,
which we implemented in the state-of-the-art tool \satsuma{}.
The key insight behind these techniques 
was to combine symmetry and cardinality reasoning.
This combination enabled us to, for the first time, 
implement meaningful, practical symmetry breaking 
that produces \sr{} proofs instead of dominance-based proofs.
Our experiments demonstrate significant performance improvements across a wide variety of benchmarks, with reduced regression compared to the lex-leader approach.

As for future work, we hope to expand the present techniques further
to close the gap with the lex-leader approach on unsatisfiable instances.
While this may be difficult by using only unit clauses,
it would be interesting to achieve this without incurring any additional
performance regression on satisfiable instances
while still using lightweight \sr{} proofs.
Indeed, it would be intriguing to find concrete benchmark families
that \emph{can} be efficiently solved in practice using lex-leader constraints,
but evade serious attempts at efficient, practical symmetry handling in \sr{}.
A concrete candidate might be small Ramsey numbers: 
while a very short \sr{} proof for $R(4, 4, 18)$ is known~\cite{CodelAH24},
it is possible to easily solve the slightly harder~$R(3, 7, 23)$ within few minutes using lex-leader constraints.
We wonder if it is possible to generalize the symmetry breaking from the $R(4, 4, 18)$ proof to larger numbers.
Separately, in MIP, orbitopal fixing can be used under more general conditions~\cite{KaibelPeinhardtPfetsch11,bendotti2021}, 
which may be interesting to adapt to SAT.

\section{Data Availability and Reproducibility}

All of our source code, benchmark formulas, and experimental results
can be found at our artifact on Zenodo.\footnote{
    \url{https://doi.org/10.5281/zenodo.17491222}.
}
The artifact contains detailed instructions
for how to download the formulas
and compile the source code in a Docker container.
The artifact also contains scripts to run experiments
and compare the results to the ones from our paper.

\section*{Acknowledgements}
We thank Christopher Hojny for pointing out important literature in MIP. 
This work has benefited substantially from
\href{https://www.dagstuhl.de/seminars/seminar-calendar/seminar-details/25231}{Dagstuhl Seminar 25231}
``Certifying Algorithms for Automated Reasoning.''
This work was supported by the National Science Foundation (NSF) under grant CCF-2415773 and funding from AFRL and DARPA under Agreement FA8750-24-9-1000.

\bibliography{main}

@article{KaibelPeinhardtPfetsch11,
title = {Orbitopal fixing},
journal = {Discrete Optimization},
volume = {8},
number = {4},
pages = {595-610},
year = {2011},
issn = {1572-5286},
doi = {https://doi.org/10.1016/j.disopt.2011.07.001},
url = {https://www.sciencedirect.com/science/article/pii/S1572528611000430},
author = {Volker Kaibel and Matthias Peinhardt and Marc E. Pfetsch},
}

@article{McConnellMNS11,
  author       = {Ross M. McConnell and
                  Kurt Mehlhorn and
                  Stefan N{\"{a}}her and
                  Pascal Schweitzer},
  title        = {Certifying algorithms},
  journal      = {Comput. Sci. Rev.},
  volume       = {5},
  number       = {2},
  pages        = {119--161},
  year         = {2011},
  url          = {https://doi.org/10.1016/j.cosrev.2010.09.009},
  doi          = {10.1016/J.COSREV.2010.09.009},
  bibsource    = {dblp computer science bibliography, https://dblp.org}
}

@inproceedings{Luks91,
  author       = {Eugene M. Luks},
  title        = {Permutation Groups and Polynomial-Time Computation},
  booktitle    = {Groups And Computation, Proceedings of a {DIMACS} Workshop, New Brunswick,
                  New Jersey, USA, October 7-10, 1991},
  series       = {{DIMACS} Series in Discrete Mathematics and Theoretical Computer Science},
  volume       = {11},
  pages        = {139--175},
  publisher    = {{DIMACS/AMS}},
  year         = {1991},
  url          = {https://doi.org/10.1090/dimacs/011/11},
  doi          = {10.1090/DIMACS/011/11},
  bibsource    = {dblp computer science bibliography, https://dblp.org}
}

@article{AloulRMS03,
  author       = {Fadi A. Aloul and
                  Arathi Ramani and
                  Igor L. Markov and
                  Karem A. Sakallah},
  title        = {Solving difficult instances of Boolean satisfiability in the presence
                  of symmetry},
  journal      = {{IEEE} Trans. Comput. Aided Des. Integr. Circuits Syst.},
  volume       = {22},
  number       = {9},
  pages        = {1117--1137},
  year         = {2003},
  url          = {https://doi.org/10.1109/TCAD.2003.816218},
  doi          = {10.1109/TCAD.2003.816218},
  bibsource    = {dblp computer science bibliography, https://dblp.org}
}

@article{HojnyP19,
  author       = {Christopher Hojny and
                  Marc E. Pfetsch},
  title        = {Polytopes associated with symmetry handling},
  journal      = {Math. Program.},
  volume       = {175},
  number       = {1-2},
  pages        = {197--240},
  year         = {2019},
  url          = {https://doi.org/10.1007/s10107-018-1239-7},
  doi          = {10.1007/S10107-018-1239-7},
  bibsource    = {dblp computer science bibliography, https://dblp.org}
}

@article{bendotti2021,
  title={Orbitopal fixing for the full (sub-)orbitope and application to the unit commitment problem},
  author={Bendotti, Pascale and Fouilhoux, Pierre and Rottner, C{\'e}cile},
  journal={Mathematical Programming},
  volume={186},
  number={1},
  pages={337--372},
  year={2021},
  publisher={Springer}
}

@misc{mexi2025,
      title={State-of-the-art Methods for Pseudo-Boolean Solving with {SCIP}}, 
      author={Gioni Mexi 
              and Dominik Kamp 
              and Yuji Shinano
              and Shanwen Pu
              and Alexander Hoen
              and Ksenia Bestuzheva
              and Christopher Hojny
              and Matthias Walter
              and Marc E. Pfetsch
              and Sebastian Pokutta
              and Thorsten Koch},
      year={2025},
      eprint={2501.03390},
      archivePrefix={arXiv},
      primaryClass={math.OC},
      url={https://arxiv.org/abs/2501.03390}, 
}

@book{handbook,
  editor       = {Armin Biere and
                  Marijn Heule and
                  Hans van Maaren and
                  Toby Walsh},
  title        = {Handbook of Satisfiability - Second Edition},
  series       = {Frontiers in Artificial Intelligence and Applications},
  volume       = {336},
  publisher    = {{IOS} Press},
  year         = {2021},
  url          = {https://doi.org/10.3233/FAIA336},
  doi          = {10.3233/FAIA336},
  isbn         = {978-1-64368-160-3},
  bibsource    = {dblp computer science bibliography, https://dblp.org}
}

@inproceedings{CrawfordGLR96,
  author    = {James M. Crawford and
               Matthew L. Ginsberg and
               Eugene M. Luks and
               Amitabha Roy},
  title     = {Symmetry-Breaking Predicates for Search Problems},
  booktitle = {Proceedings of the Fifth International Conference on Principles of
               Knowledge Representation and Reasoning (KR'96), Cambridge, Massachusetts,
               USA, November 5-8, 1996},
  pages     = {148--159},
  publisher = {Morgan Kaufmann},
  year      = {1996},
  bibsource = {dblp computer science bibliography, https://dblp.org}
}

@inproceedings{AloulMS03,
  author       = {Fadi A. Aloul and
                  Igor L. Markov and
                  Karem A. Sakallah},
  title        = {Shatter: efficient symmetry-breaking for boolean satisfiability},
  booktitle    = {Proceedings of the 40th Design Automation Conference, {DAC} 2003,
                  Anaheim, CA, USA, June 2-6, 2003},
  pages        = {836--839},
  publisher    = {{ACM}},
  year         = {2003},
  url          = {https://doi.org/10.1145/775832.776042},
  doi          = {10.1145/775832.776042},
  bibsource    = {dblp computer science bibliography, https://dblp.org}
}

@article{JunttilaKKK20,
  author       = {Tommi A. Junttila and
                  Matti Karppa and
                  Petteri Kaski and
                  Jukka Kohonen},
  title        = {An adaptive prefix-assignment technique for symmetry reduction},
  journal      = {J. Symb. Comput.},
  volume       = {99},
  pages        = {21--49},
  year         = {2020},
  url          = {https://doi.org/10.1016/j.jsc.2019.03.002},
  doi          = {10.1016/J.JSC.2019.03.002},
  bibsource    = {dblp computer science bibliography, https://dblp.org}
}

@inproceedings{ShengRH25,
  author       = {Aeacus Sheng and
                  Joseph E. Reeves and
                  Marijn J. H. Heule},
  title        = {Reencoding Unique Literal Clauses},
  booktitle    = {28th International Conference on Theory and Applications of Satisfiability
                  Testing, {SAT} 2025, August 12-15, 2025, Glasgow, Scotland},
  series       = {LIPIcs},
  volume       = {341},
  pages        = {29:1--29:21},
  publisher    = {Schloss Dagstuhl - Leibniz-Zentrum f{\"{u}}r Informatik},
  year         = {2025},
  url          = {https://doi.org/10.4230/LIPIcs.SAT.2025.29},
  doi          = {10.4230/LIPICS.SAT.2025.29},
  bibsource    = {dblp computer science bibliography, https://dblp.org}
}

@inproceedings{AndersBR24,
  author       = {Markus Anders and
                  Sofia Brenner and
                  Gaurav Rattan},
  title        = {Satsuma: Structure-Based Symmetry Breaking in {SAT}},
  booktitle    = {27th International Conference on Theory and Applications of Satisfiability
                  Testing, {SAT} 2024, August 21-24, 2024, Pune, India},
  series       = {LIPIcs},
  volume       = {305},
  pages        = {4:1--4:23},
  publisher    = {Schloss Dagstuhl - Leibniz-Zentrum f{\"{u}}r Informatik},
  year         = {2024},
  url          = {https://doi.org/10.4230/LIPIcs.SAT.2024.4},
  doi          = {10.4230/LIPICS.SAT.2024.4},
  bibsource    = {dblp computer science bibliography, https://dblp.org}
}

@inproceedings{DevriendtBBD16,
  author       = {Jo Devriendt and
                  Bart Bogaerts and
                  Maurice Bruynooghe and
                  Marc Denecker},
  title        = {Improved Static Symmetry Breaking for {SAT}},
  booktitle    = {Theory and Applications of Satisfiability Testing - {SAT} 2016 - 19th
                  International Conference, Bordeaux, France, July 5-8, 2016, Proceedings},
  series       = {Lecture Notes in Computer Science},
  volume       = {9710},
  pages        = {104--122},
  publisher    = {Springer},
  year         = {2016},
  url          = {https://doi.org/10.1007/978-3-319-40970-2\_8},
  doi          = {10.1007/978-3-319-40970-2\_8},
  bibsource    = {dblp computer science bibliography, https://dblp.org}
}

@article{BogaertsGMN23,
  author       = {Bart Bogaerts and
                  Stephan Gocht and
                  Ciaran McCreesh and
                  Jakob Nordstr{\"{o}}m},
  title        = {Certified Dominance and Symmetry Breaking for Combinatorial Optimisation},
  journal      = {J. Artif. Intell. Res.},
  volume       = {77},
  pages        = {1539--1589},
  year         = {2023},
  url          = {https://doi.org/10.1613/jair.1.14296},
  doi          = {10.1613/JAIR.1.14296},
  bibsource    = {dblp computer science bibliography, https://dblp.org}
}

@article{Sabharwal09,
  author       = {Ashish Sabharwal},
  title        = {{SymChaff}: exploiting symmetry in a structure-aware satisfiability
                  solver},
  journal      = {Constraints An Int. J.},
  volume       = {14},
  number       = {4},
  pages        = {478--505},
  year         = {2009},
  url          = {https://doi.org/10.1007/s10601-008-9060-1},
  doi          = {10.1007/S10601-008-9060-1},
  bibsource    = {dblp computer science bibliography, https://dblp.org}
}

@inproceedings{DevriendtBCDM12,
  author       = {Jo Devriendt and
                  Bart Bogaerts and
                  Broes De Cat and
                  Marc Denecker and
                  Christopher Mears},
  title        = {Symmetry Propagation: Improved Dynamic Symmetry Breaking in {SAT}},
  booktitle    = {{IEEE} 24th International Conference on Tools with Artificial Intelligence,
                  {ICTAI} 2012, Athens, Greece, November 7-9, 2012},
  pages        = {49--56},
  publisher    = {{IEEE} Computer Society},
  year         = {2012},
  url          = {https://doi.org/10.1109/ICTAI.2012.16},
  doi          = {10.1109/ICTAI.2012.16},
  bibsource    = {dblp computer science bibliography, https://dblp.org}
}

@inproceedings{Devriendt0B17,
  author       = {Jo Devriendt and
                  Bart Bogaerts and
                  Maurice Bruynooghe},
  title        = {Symmetric Explanation Learning: Effective Dynamic Symmetry Handling
                  for {SAT}},
  booktitle    = {Theory and Applications of Satisfiability Testing - {SAT} 2017 - 20th
                  International Conference, Melbourne, VIC, Australia, August 28 - September
                  1, 2017, Proceedings},
  series       = {Lecture Notes in Computer Science},
  volume       = {10491},
  pages        = {83--100},
  publisher    = {Springer},
  year         = {2017},
  url          = {https://doi.org/10.1007/978-3-319-66263-3\_6},
  doi          = {10.1007/978-3-319-66263-3\_6},
  bibsource    = {dblp computer science bibliography, https://dblp.org}
}

@incollection{GentPP06,
  author       = {Ian P. Gent and
                  Karen E. Petrie and
                  Jean{-}Fran{\c{c}}ois Puget},
  editor       = {Francesca Rossi and
                  Peter van Beek and
                  Toby Walsh},
  title        = {Symmetry in Constraint Programming},
  booktitle    = {Handbook of Constraint Programming},
  series       = {Foundations of Artificial Intelligence},
  volume       = {2},
  pages        = {329--376},
  publisher    = {Elsevier},
  year         = {2006},
  url          = {https://doi.org/10.1016/S1574-6526(06)80014-3},
  doi          = {10.1016/S1574-6526(06)80014-3},
  bibsource    = {dblp computer science bibliography, https://dblp.org}
}

@article{PfetschR19,
  author       = {Marc E. Pfetsch and
                  Thomas Rehn},
  title        = {A computational comparison of symmetry handling methods for mixed
                  integer programs},
  journal      = {Math. Program. Comput.},
  volume       = {11},
  number       = {1},
  pages        = {37--93},
  year         = {2019},
  url          = {https://doi.org/10.1007/s12532-018-0140-y},
  doi          = {10.1007/S12532-018-0140-Y},
  bibsource    = {dblp computer science bibliography, https://dblp.org}
}

@inproceedings{CodelAH24,
  author       = {Cayden R. Codel and
                  Jeremy Avigad and
                  Marijn J. H. Heule},
  title        = {Verified Substitution Redundancy Checking},
  booktitle    = {Formal Methods in Computer-Aided Design, {FMCAD} 2024, Prague, Czech
                  Republic, October 15-18, 2024},
  pages        = {186--196},
  publisher    = {{IEEE}},
  year         = {2024},
  url          = {https://doi.org/10.34727/2024/isbn.978-3-85448-065-5\_24},
  doi          = {10.34727/2024/ISBN.978-3-85448-065-5\_24},
  bibsource    = {dblp computer science bibliography, https://dblp.org}
}

@inproceedings{GochtN21,
  author       = {Stephan Gocht and
                  Jakob Nordstr{\"{o}}m},
  title        = {Certifying Parity Reasoning Efficiently Using Pseudo-Boolean Proofs},
  booktitle    = {Thirty-Fifth {AAAI} Conference on Artificial Intelligence, {AAAI}
                  2021, Thirty-Third Conference on Innovative Applications of Artificial
                  Intelligence, {IAAI} 2021, The Eleventh Symposium on Educational Advances
                  in Artificial Intelligence, {EAAI} 2021, Virtual Event, February 2-9,
                  2021},
  pages        = {3768--3777},
  publisher    = {{AAAI} Press},
  year         = {2021},
  url          = {https://doi.org/10.1609/aaai.v35i5.16494},
  doi          = {10.1609/AAAI.V35I5.16494},
  bibsource    = {dblp computer science bibliography, https://dblp.org}
}

@incollection{SIMS1970169,
    author = {Charles C. Sims},
    title = {Computational methods in the study of permutation groups},
    booktitle = {Computational Problems in Abstract Algebra},
    publisher = {Pergamon},
    pages = {169-183},
    year = {1970},
    isbn = {978-0-08-012975-4},
    doi = {https://doi.org/10.1016/B978-0-08-012975-4.50020-5},
    url = {https://www.sciencedirect.com/science/article/pii/B9780080129754500205},
}

@book{seress_2003, place={Cambridge}, series={Cambridge Tracts in Mathematics}, title={Permutation Group Algorithms}, DOI={10.1017/CBO9780511546549}, publisher={Cambridge University Press}, author={Seress, Ákos}, year={2003}, collection={Cambridge Tracts in Mathematics}}

@inproceedings{AndersS21,
  author       = {Markus Anders and
                  Pascal Schweitzer},
  title        = {Parallel Computation of Combinatorial Symmetries},
  booktitle    = {29th Annual European Symposium on Algorithms, {ESA} 2021, September
                  6-8, 2021, Lisbon, Portugal (Virtual Conference)},
  series       = {LIPIcs},
  volume       = {204},
  pages        = {6:1--6:18},
  publisher    = {Schloss Dagstuhl - Leibniz-Zentrum f{\"{u}}r Informatik},
  year         = {2021},
  url          = {https://doi.org/10.4230/LIPIcs.ESA.2021.6},
  doi          = {10.4230/LIPICS.ESA.2021.6},
  bibsource    = {dblp computer science bibliography, https://dblp.org}
}

@inproceedings{AndersS21b,
  author       = {Markus Anders and
                  Pascal Schweitzer},
  title        = {Search Problems in Trees with Symmetries: Near Optimal Traversal Strategies
                  for Individualization-Refinement Algorithms},
  booktitle    = {48th International Colloquium on Automata, Languages, and Programming,
                  {ICALP} 2021, July 12-16, 2021, Glasgow, Scotland (Virtual Conference)},
  series       = {LIPIcs},
  volume       = {198},
  pages        = {16:1--16:21},
  publisher    = {Schloss Dagstuhl - Leibniz-Zentrum f{\"{u}}r Informatik},
  year         = {2021},
  url          = {https://doi.org/10.4230/LIPIcs.ICALP.2021.16},
  doi          = {10.4230/LIPICS.ICALP.2021.16},
  bibsource    = {dblp computer science bibliography, https://dblp.org}
}

@article{McKayP14,
  author       = {Brendan D. McKay and
                  Adolfo Piperno},
  title        = {Practical graph isomorphism, {II}},
  journal      = {J. Symb. Comput.},
  volume       = {60},
  pages        = {94--112},
  year         = {2014},
  url          = {https://doi.org/10.1016/j.jsc.2013.09.003},
  doi          = {10.1016/J.JSC.2013.09.003},
  bibsource    = {dblp computer science bibliography, https://dblp.org}
}

@inproceedings{practicaliso1,
	author = {Brendan D. McKay},
	title = {Practical Graph Isomorphism},
	booktitle = {10th. Manitoba Conference on Numerical Mathematics and Computing (Winnipeg, 1980)},
	pages = {45--87},
	year = {1981}
}

@article{corr/abs-0804-4881,
  author       = {Adolfo Piperno},
  title        = {Search Space Contraction in Canonical Labeling of Graphs (Preliminary
                  Version)},
  journal      = {CoRR},
  volume       = {abs/0804.4881},
  year         = {2008},
  url          = {http://arxiv.org/abs/0804.4881},
  eprinttype    = {arXiv},
  eprint       = {0804.4881},
  bibsource    = {dblp computer science bibliography, https://dblp.org}
}

@inproceedings{JunttilaK07,
  author       = {Tommi A. Junttila and
                  Petteri Kaski},
  title        = {Engineering an Efficient Canonical Labeling Tool for Large and Sparse
                  Graphs},
  booktitle    = {Proceedings of the Nine Workshop on Algorithm Engineering and Experiments,
                  {ALENEX} 2007, New Orleans, Louisiana, USA, January 6, 2007},
  publisher    = {{SIAM}},
  year         = {2007},
  url          = {https://doi.org/10.1137/1.9781611972870.13},
  doi          = {10.1137/1.9781611972870.13},
  bibsource    = {dblp computer science bibliography, https://dblp.org}
}

@inproceedings{JunttilaK11,
  author       = {Tommi A. Junttila and
                  Petteri Kaski},
  title        = {Conflict Propagation and Component Recursion for Canonical Labeling},
  booktitle    = {Theory and Practice of Algorithms in (Computer) Systems - First International
                  {ICST} Conference, {TAPAS} 2011, Rome, Italy, April 18-20, 2011. Proceedings},
  series       = {Lecture Notes in Computer Science},
  volume       = {6595},
  pages        = {151--162},
  publisher    = {Springer},
  year         = {2011},
  url          = {https://doi.org/10.1007/978-3-642-19754-3\_16},
  doi          = {10.1007/978-3-642-19754-3\_16},
  bibsource    = {dblp computer science bibliography, https://dblp.org}
}

@inproceedings{FlenerFHKMPW01,
  author       = {Pierre Flener and
                  Alan M. Frisch and
                  Brahim Hnich and
                  Zeynep Kiziltan and
                  Ian Miguel and
                  Justin Pearson and
                  Toby Walsh},
  title        = {Breaking Row and Column Symmetries in Matrix Models},
  booktitle    = {Principles and Practice of Constraint Programming - {CP} 2002, 8th
                  International Conference, {CP} 2002, Ithaca, NY, USA, September 9-13,
                  2002, Proceedings},
  series       = {Lecture Notes in Computer Science},
  volume       = {2470},
  pages        = {462--476},
  publisher    = {Springer},
  year         = {2002},
  url          = {https://doi.org/10.1007/3-540-46135-3\_31},
  doi          = {10.1007/3-540-46135-3\_31},
  bibsource    = {dblp computer science bibliography, https://dblp.org}
}

@article{BussHierarchy,
  title = {{DRAT} and Propagation Redundancy Proofs Without New Variables},
  author = {Sam Buss and Neil Thapen},
  url = {https://lmcs.episciences.org/7400},
  doi = {10.23638/LMCS-17(2:12)2021},
  journal = {{Logical Methods in Computer Science}},
  volume = {Volume 17, Issue 2},
  year = {2021},
  month = {Apr},
}

@InProceedings{RebolaPardo23,
  author =	{Rebola-Pardo, Adri\'{a}n},
  title =	{Even Shorter Proofs Without New Variables},
  booktitle =	{26th International Conference on Theory and Applications of Satisfiability Testing (SAT 2023)},
  pages =	{22:1--22:20},
  series =	{Leibniz International Proceedings in Informatics (LIPIcs)},
  ISBN =	{978-3-95977-286-0},
  ISSN =	{1868-8969},
  year =	{2023},
  volume =	{271},
  editor =	{Mahajan, Meena and Slivovsky, Friedrich},
  publisher =	{Schloss Dagstuhl -- Leibniz-Zentrum f{\"u}r Informatik},
  address =	{Dagstuhl, Germany},
  URL =		{https://drops.dagstuhl.de/entities/document/10.4230/LIPIcs.SAT.2023.22},
  URN =		{urn:nbn:de:0030-drops-184844},
  doi =		{10.4230/LIPIcs.SAT.2023.22},
}

@InProceedings{RAT12,
  author = {J{\"a}rvisalo, Matti
            and Heule, Marijn J. H.
            and Biere, Armin},
  IGNOREeditor = {Gramlich, Bernhard
            and Miller, Dale
            and Sattler, Uli},
  title = {Inprocessing Rules},
  booktitle = {Automated Reasoning},
  year = {2012},
  IGNOREpublisher = {Springer Berlin Heidelberg},
  IGNOREaddress = {Berlin, Heidelberg},
  pages = {355--370},
  isbn = {978-3-642-31365-3},
}

@article{PRTheory,
  title   = {Strong Extension-Free Proof Systems},
  author  = {Marijn J. H. Heule and
             Benjamin Kiesl and
             Armin Biere},
  journal = {Journal of Automated Reasoning},
  volume  = {64},
  number  = {3},
  year    = {2020},
  pages   = {533--554},
  doi     = {10.1007/s10817-019-09516-0},
}

@inproceedings{Cadical2024,
  author       = {Armin Biere and
                  Tobias Faller and
                  Katalin Fazekas and
                  Mathias Fleury and
                  Nils Froleyks and
                  Florian Pollitt},
  editor       = {Arie Gurfinkel and
                  Vijay Ganesh},
  title        = {{CaDiCaL 2.0}},
  booktitle    = {Computer Aided Verification - 36th International Conference, {CAV}
                  2024, Montreal, QC, Canada, July 24-27, 2024, Proceedings, Part {I}},
  series       = {Lecture Notes in Computer Science},
  volume       = {14681},
  pages        = {133--152},
  publisher    = {Springer},
  year         = {2024},
  doi          = {10.1007/978-3-031-65627-9\_7},
}

@proceedings{satcomp2022,
  editor = {Tomas Balyo and
            Marijn J. H. Heule and
            Markus Iser and
            Matti J{\"a}rvisalo and
            Martin Suda},
  published = {},
  publisher = {Helsinki Institute for Information Technology},
  series = {Department of Computer Science Series of Publications B},
  title = {Proceedings of SAT Competition 2022: Solver and Benchmark Descriptions},
  volume = {},
  year = {2022}
}

@inproceedings{bridges2,
  author = {Brown, Shawn T.
            and Buitrago, Paola
            and Hanna, Edward
            and Sanielevici, Sergiu
            and Scibek, Robin
            and Nystrom, Nicholas A.},
  title = {Bridges-2: A Platform for Rapidly-Evolving and Data Intensive Research},
  year = {2021},
  isbn = {9781450382922},
  publisher = {Association for Computing Machinery},
  address = {New York, NY, USA},
  doi = {10.1145/3437359.3465593},
  booktitle = {Practice and Experience in Advanced Research Computing 2021: Evolution Across All Dimensions},
  articleno = {35},
  numpages = {4},
  location = {Boston, MA, USA},
  series = {PEARC '21}
}
\bibliographystyle{plain}

\end{document}